\documentclass[letterpaper, 10 pt, conference]{ieeeconf}  % Comment this line out if you need a4paper

\IEEEoverridecommandlockouts                              % This command is only needed if 
\usepackage[natbib=true,backend=biber,style=numeric-comp,sorting=none,maxbibnames=99]{biblatex}

\usepackage{graphicx,algorithm,algpseudocode, amsthm,amsmath,amsfonts,verbatim,mathtools,enumerate,bm,hyperref,dsfont,xcolor, caption, subcaption, rotating, tikz, pgfplots}
\usetikzlibrary{shapes, positioning}
\usepackage[export]{adjustbox}
\definecolor{preyred}{rgb}{0.8392156862745098, 0.15294117647058825, 0.1568627450980392}
\definecolor{predatorblue}{rgb}{0.12156862745098039, 0.4666666666666667, 0.7058823529411765}
\usepackage[capitalize]{cleveref}

\allowdisplaybreaks

\newcommand*\octagoned[3]{\tikz[baseline=(char.base)]{
    \node[regular polygon, regular polygon sides=8, fill=#2, text=#3, inner sep=2pt] (char) {#1};}}
\newcommand{\diag}{\mathop{\rm diag}}
\newcommand{\blkdiag}{\mathop{\rm blkdiag}}

\newcommand{\rem}{\mathop{\rm rem}}
\newcommand{\quot}{\mathop{\rm quo}}
\newcommand{\vect}{\mathop{\rm vec}}

\newcommand{\argmin}{\mathop{\rm argmin}}

\newcommand{\norm}[1]{\left\lVert#1\right\rVert}
\newcommand{\mnorm}[1]{{\left\vert\kern-0.25ex\left\vert\kern-0.25ex\left\vert #1 
    \right\vert\kern-0.25ex\right\vert\kern-0.25ex\right\vert}}

\theoremstyle{definition}
\newtheorem{definition}{Definition} 
\newtheorem{theorem}{Theorem}

\usepackage{soul}

\bibliography{reference}

\DeclareCaptionFormat{myformat}{\small\itshape #1#2#3}
\title{\LARGE \bf Soft-Bellman Equilibrium in Affine Markov Games:\\ Forward Solutions and Inverse Learning
}

\author{Shenghui~Chen, Yue~Yu, David~Fridovich-Keil, and Ufuk~Topcu% <-this % stops a space
\thanks{S. Chen, Y. Yu, D. Fridovich-Keil, and U. Topcu are with the Oden Institute for Computational Engineering and Sciences, The University of Texas at Austin, TX, 78712, USA (emails: shenghui.chen@utexas.edu,\, yueyu@utexas.edu,\, dfk@utexas.edu,\, utopcu@utexas.edu).}%
}

\begin{document}

\maketitle
\thispagestyle{empty}
\pagestyle{empty}

%%%%%%%%%%%%%%%%%%%%%%%%%%%%%%%%%%%%%%%%%%%%%%%%%%%%%%%%%%%%%%%%%%%%%%%%%%%%%%%%
\begin{abstract}
Markov games model interactions among multiple players in a stochastic, dynamic environment. Each player in a Markov game maximizes its expected total discounted reward, which depends upon the policies of the other players. We formulate a class of Markov games, termed \emph{affine Markov games}, where an affine reward function couples the players' actions. We introduce a novel solution concept, the \emph{soft-Bellman equilibrium}, where each player is boundedly rational and chooses a soft-Bellman policy rather than a purely rational policy as in the well-known Nash equilibrium concept. We provide conditions for the existence and uniqueness of the soft-Bellman equilibrium and propose a nonlinear least-squares algorithm to compute such an equilibrium in the \emph{forward problem}. We then solve the \emph{inverse game problem} of inferring the players’ reward parameters from observed state-action trajectories via a projected-gradient algorithm. Experiments in a predator-prey OpenAI Gym environment show that the reward parameters inferred by the proposed algorithm outperform those inferred by a baseline algorithm: they reduce the Kullback-Leibler divergence between the equilibrium policies and observed policies by at least two orders of magnitude.
\end{abstract}

%%%%%%%%%%%%%%%%%%%%%%%%%%%%%%%%%%%%%%%%%%%%%%%%%%%%%%%%%%%%%%%%%%%%%%%%%%%%%%%%
\section{Introduction}
\label{sec: introduction}

Markov games model the interaction of multiple decision makers in stochastic and dynamic environments \cite{shapley1953stochastic}. In a Markov game, each player's transition and reward depend on the policies of the other players, and each player aims to find an optimal policy that maximizes its expected discounted total reward.

% ----

% The Nash equilibrium is a fundamental concept for analyzing strategic interactions among multiple decision-makers \cite{nash1996non}. In a Markov game, it refers to a collection of policies where no player can benefit by unilaterally changing its policy \cite{shapley1953stochastic}. The Nash equilibrium concept assumes that the players are perfectly rational, seeking to maximize their rewards without perceptual errors or cognitive biases.

% However, not all players are perfectly rational. Humans, for example, have limited cognitive capacity and are subject to biases and heuristics that can affect their decision-making, bounding their rationality. As a result, the outcomes of games played by humans may not always align with the predictions from the Nash equilibrium concept.
% Recent efforts attempt to address this limitation of the Nash equilibrium by accounting for players' bounded rationality in games with specific structures, including matrix games \cite{waugh2013computational}, fully cooperative games \cite{barrett2017making, le2017coordinated, vsovsic2016inverse}, and two-player games \cite{lin2017multiagent, lin2019multi}. Another recent work tackles the same limitation in dynamic games with continuous state and action spaces \cite{mehr2023maximum}. However, to our best knowledge, no work has addressed this limitation in \emph{general-sum, multi-player} Markov games with \emph{discrete} state and action spaces yet.

The concept of Nash equilibrium, which refers to a set of policies where no player can benefit by unilaterally changing their policy \cite{shapley1953stochastic}, overlooks the reality that players are often boundedly rational. For example, humans have limited cognitive capacity and are subject to biases and heuristics that can affect their decision-making. As a result, the outcomes of games played by humans may not always align with the predictions from the Nash equilibrium concept. Recent efforts attempt to address this limitation by accounting for players' bounded rationality in games with specific structures, including matrix games \cite{waugh2013computational}, fully cooperative games \cite{barrett2017making, le2017coordinated, vsovsic2016inverse}, and two-player games \cite{lin2017multiagent, lin2019multi}. Another recent work tackles the same limitation in dynamic games with continuous state and action spaces \cite{mehr2023maximum}. However, to our best knowledge, no work has addressed this limitation in \emph{general-sum, multi-player} Markov games with \emph{discrete} state and action spaces yet.

% ----

We propose the \emph{soft-Bellman equilibrium} as a new solution concept to capture the dynamics of boundedly rational players in \emph{affine Markov games}.
Affine Markov games are a class of Markov games where each player has independent dynamics and an affine reward function couples the players’ actions.
In a soft-Bellman equilibrium, each player chooses a policy that maximizes the expected reward with causal entropy regularization while satisfying independent transition dynamics.
We provide conditions for the existence and uniqueness of the soft-Bellman equilibrium.

% ----

% forward; inverse
We study the \emph{forward problem} of computing a soft-Bellman equilibrium in a given affine Markov game.
We propose a least-squares-based algorithm to solve this problem by minimizing the residuals of the soft-Bellman equilibrium conditions.
% We propose a least-squares-based algorithm to solve a causal-entropy regularized equilibrium problem by minimizing the residuals of equilibrium conditions. 

We then turn to the \emph{inverse game problem} of inferring the players’ reward parameters that best explain observed interactions. We propose an iterative algorithm that leverages the solutions to the forward problem. In each iteration, the algorithm computes the soft-Bellman equilibrium given the current reward parameters and then updates those parameters with a projected-gradient method based on the implicit function theorem \cite{dontchev2014implicit}. 

The proposed inverse game algorithm outperforms a baseline algorithm that ignores the coupling between players. 
Experiments in a predator-prey OpenAI Gym environment \cite{magym} show that the reward parameters inferred by the proposed algorithm reduce the Kullback-Leibler divergence between the equilibrium policies and observed policies by at least two orders of magnitude than the baseline algorithm.
% Using a synthetic dataset in a predator-prey OpenAI Gym environment \cite{magym}, we compare the proposed inverse game algorithm with a baseline algorithm that ignores the coupling between players. 
% Results show that the proposed algorithm terminates with a Kullback-Leibler divergence between the equilibrium policies and observed policies at least two orders of magnitude lower than that of the baseline algorithm.
\section{Related Work}
\label{sec: related}
In single-agent settings, literature in inverse reinforcement learning studies the problem of inferring reward parameters from human experts' trajectories. The principle of maximum entropy is a popular approach in this direction \cite{ziebart2008maximum}. Subsequent studies further extend this principle to accommodate stochastic transitions using causal entropy \cite{ziebart2013principle}. For example, recent work extends the maximum causal entropy framework in inverse reinforcement learning to an infinite time horizon setting and proposes the concept of stationary soft-Bellman policy \cite{zhou2017infinite}. This policy concept inspires the formulation of the soft-Bellman equilibrium to account for the players' bounded rationality, a feature lacking in the Nash equilibrium concept.

In multi-agent settings, most existing works that try to address the limitation of Nash equilibrium assume specific game structures, including matrix games \cite{waugh2013computational}, fully cooperative games \cite{barrett2017making, le2017coordinated, vsovsic2016inverse}, two-player zero-sum games \cite{lin2017multiagent}, and two-player general-sum games \cite{lin2019multi}.
This paper generalizes the existing works to \emph{multi-player, general-sum} Markov games.

First formulated in normal-form and extensive-form games, the quantal response equilibrium is a solution concept to model the bounded rationality of human players \cite{mckelvey1995quantal, mckelvey1998quantal}.
Inspired by this solution concept, recent work proposes the entropic cost equilibrium to extend the quantal response equilibrium to games with \emph{continuous} states and actions \cite{mehr2023maximum}. 

The current work, on the other hand, proposes the soft-Bellman equilibrium to support stochastic transitions in affine Markov games with \emph{discrete} state and action spaces. 
Although both the entropic cost equilibrium and the soft-Bellman equilibrium extend the quantal response equilibrium, the soft-Bellman equilibrium is different in choosing the state-action frequency matrix, instead of the policy, as the variable to optimize for each player. This subtle difference changes the expected reward from a nonconvex function to a convex one, laying the groundwork for establishing conditions that ensure the existence and uniqueness of solutions.

% Some other works that attempt to solve the inverse game problem have different assumptions than ours. Using machine learning, \citet{song2018multi} extend the generative adversarial imitation learning to multi-agent settings based on the notion of Nash equilibrium. Also based on adversarial training, \citet{yu2019multi} propose a multi-agent IRL method for high-dimensional state and action spaces with \emph{unknown} dynamics. 
% \section{Affine Markov games and soft-Bellman equilibrium}
% We introduce a class of Markov games, termed \emph{affine Markov games}, along with the notion of soft-Bellman equilibrium.

\section{Models}
We present our main theoretical models: a special class of Markov games, along with a novel equilibrium concept that accounts for bounded rationality.

\subsection{Affine Markov Games}
We consider a Markov game \cite{shapley1953stochastic} where each player solves an MDP with independent dynamics and an affine reward function that couples the players' actions. We let \(p\in\mathbb{N}\) denote the number of players. Player \(i\in[p]\) solves an MDP specified by a tuple which includes a set of states, a set of actions, a transition kernel, an initial state distribution, a reward matrix, and a discount factor. We let \(n^i\in\mathbb{N}\) and \(m^i\in\mathbb{N}\) denote the number of states and actions for player \(i\), respectively. We let \(S_t^i\in[n^i]\) and \(A_t^i\in[m^i]\) denote the state and action of player \(i\) at time \(t\in\mathbb{N}\). Each action triggers a stochastic transition between the current state to the next state. We let \(T^i\in\mathbb{R}^{n^i\times m^i\times n^i}\) denote the \emph{transition kernel} of player \(i\) such that
\begin{equation}\label{eqn: transition kernel}
    T_{saj}^i\coloneqq \mathds{P}(S_{t+1}^i=j|S_t^i=s, A_t^i=a)
\end{equation}
for all \(t\in\mathbb{N}\), \(s, j\in[n^i]\) and \(a\in[m^i]\). We let \(q^i\in\mathbb{R}^{n^i}\) denote the \emph{initial state distribution} of player \(i\) such that
\begin{equation}\label{eqn: init dist}
    q_s^i\coloneqq \mathds{P}(S_0^i=s)
\end{equation}
for all \(s\in[n^i]\). We let \(R^i\in\mathbb{R}^{n^i\times m^i}\) denote the reward matrix, where
\(R_{sa}^i\) denotes the reward of player \(i\) for choosing choosing action \(a\) in state \(s\). Finally, we let \(\gamma\in[0, 1)\) denote a reward discount factor. For each player \(i\in[p]\), a {stationary policy} maps each state to a probability distribution over actions. We denote such a policy as a matrix \(\Pi^i\in\mathbb{R}^{n^i\times m^i}\) where
\begin{equation}\label{eqn: policy}
     \Pi_{sa}^i\coloneqq \mathds{P}(A_t^i=a| S_t^i=s)
\end{equation}
for all \(t\in\mathbb{N}, s\in[n^i], a\in[m^i]\). An optimal stationary policy in an MDP minimizes the following expected total discounted state-action reward
\begin{equation}\label{eqn: reward sum}
    \sum_{t=0}^\infty \sum_{s=1}^n\sum_{a=1}^m\gamma^t \mathds{P}(S_t^i=s, A_t^i=a) R_{sa}^i.
\end{equation}
We let \(Y^i\in\mathbb{R}^{n^i\times m^i}\) denote the state-action frequency matrix of player \(i\in[p]\) such that
\begin{equation}\label{eqn: state-action frequency}
    Y_{sa}^i\coloneqq\sum_{t=0}^\infty \gamma^t \mathds{P}(S_t^i=s, A_t^i=a).
\end{equation}
for all \(s\in[n^i]\) and \(a\in[m^i]\). 

We now introduce the definition of a \(p\)-player affine Markov game. 

\begin{definition}
     A \(p\)-player affine Markov game is a collection of MDPs \(\{\mathcal{M}^i=\{[n^i], [m^i], q^i, T^i, R^i, \gamma\}\}_{i=1}^p\) such that there exists \(b^i\in\mathbb{R}^{m^in^i}\) and \(C^{ij}\in\mathbb{R}^{m^in^i\times m^jn^j}\) for each \(i, j\in[p]\) such that 
\begin{equation}\label{eqn: game reward}
\vect(R^i)=b^i+\sum_{j=1}^p C^{ij} \vect(Y^j)
\end{equation}
for all \(i\in[p]\), where \(Y^i\in\mathbb{R}^{m^i\times n^i}\) satisfies \eqref{eqn: state-action frequency}.
\end{definition}

The affine reward structure in \eqref{eqn: game reward} couples different players' decisions together: the reward for a player is not a fixed number, but depends on the other players' state-action frequencies. Similar coupling appears in matrix games where each player has a finite number of candidate options \cite{rosen1965existence,yu2022inverse}. This function has two parameters: $b$ pertains to the individual player, and $C$ considers the coupling between the players.

\subsection{Soft-Bellman Equilibrium}

We now introduce the notion of \emph{soft-Bellman  equilibrium}. 
% It extends the quantal response equilibrium in static games to dynamic games. 
It extends the notion of quantal response equilibrium in games with deterministic dynamics to Markov games with stochastic dynamics \cite{mckelvey1995quantal, mckelvey1998quantal}.
Unlike Nash equilibrium, it states that all players choose a soft-Bellman policy---rather than the optimal policy that satisfies the Bellman equations---given other players' actions.

\begin{definition}\label{def: soft-Bellman equilibrium}
Let \(\{\mathcal{M}^i=\{[n^i], [m^i], q^i, T^i, R^i, \gamma\}\}_{i=1}^p\) be an affine Markov game. Let \(\Pi^i\in\mathbb{R}^{n^i\times m^i}\) 
be a stationary policy matrix of player \(i\in[p]\). If there exists \(v^i\in\mathbb{R}^{n^i}\) and \(Q^i\in\mathbb{R}^{n^i\times m^i}\) such that
\begin{subequations}
\begin{align}
    \Pi_{sa}^i&=\frac{\exp( Q_{sa}^i)}{\sum_{j=1}^{m^i}\exp(Q_{sj}^i)},\label{eqn: soft policy} \\
Q_{sa}^i&=R_{sa}^i+\gamma \sum_{j=1}^{n^i} T_{saj}^iv_j^i,\label{eqn: soft Q} \\
v_s^i &= \log\left(\sum_{a=1}^{m^i} \exp(Q_{sa}^i)\right),\label{eqn: soft V}
\end{align}
\end{subequations}
for all \(s\in[n^i]\) and \(a\in[m^i]\), then \(\{\Pi^i\}_{i=1}^p\), is a \emph{soft-Bellman equilibrium} for \(\{\mathcal{M}^i\}_{i=1}^p\).  
\end{definition}

Previous studies have proposed similar notions of equilibrium, e.g., Markov quantal response equilibrium in \cite{eibelshauser2019markov}. However, unlike \Cref{eqn: soft V} in \cref{def: soft-Bellman equilibrium}, their formulations are inconsistent with the characterization of soft-Bellman policies. There is a close connection between the soft-Bellman equilibrium and the following optimization over state-action frequency matrix:
\begin{equation}\label{opt: best response}
    \begin{array}{ll}
 \underset{Y\in\mathbb{R}^{n^i\times m^i}}{\mbox{maximize}}  & \ell^i(Y)+h(Y)\\
      \mbox{subject to}  & \sum\limits_{a=1}^{m^i}Y_{sa}=q_s^i + \gamma \sum\limits_{j=1}^{n^i}\sum\limits_{a=1}^{m^i} T_{jas}^i Y_{ja},\enskip s\in[n^i],
      \end{array}
\end{equation} 
where
\begin{subequations}
    \begin{align}
        &\begin{aligned}
\ell^i(Y) &\coloneqq \vect(Y)^\top b^i+\frac{1}{2}\vect(Y)^\top C^{ii} \vect(Y)\\
&+\sum_{\substack{j=1, j\neq i}}^p \vect(Y^j)^\top C^{ij}\vect(Y),
\end{aligned} \label{eqn: potential}\\
 &h(Y)\coloneqq\sum_{s=1}^n\sum_{a=1}^mY_{sa}\left(\log\left(\sum_{j=1}^m Y_{sj}\right)-\log(Y_{sa})\right).\label{eqn: entropy} 
    \end{align}
\end{subequations}

The following theorem shows that, if each player chooses its policy by solving optimization~\eqref{opt: best response}, then the resulting policies form a soft-Bellman  equilibrium.  
\begin{theorem}\label{lem: soft-Bellman equilibrium}
Let \(\{\mathcal{M}^i=\{[n^i], [m^i], q^i, T^i, R^i, \gamma\}\}_{i=1}^p\) be an affine Markov game. Suppose that \(C^{ii}\preceq 0\) and \(Y^i\in\mathbb{R}_{>0}^{n^i\times m^i}\) is an optimal solution of optimization~\eqref{opt: best response} for all \(i\in[p]\). Let \(\Pi^i\in\mathbb{R}^{n^i\times m^i}\) be such that 
\begin{equation}\label{eqn: Y2Pi}
\Pi_{sa}^i=
\frac{Y_{sa}^i}{\sum_{j=1}^{m^i} Y_{sj}^i}
\end{equation}
for all \(i\in[p]\), \(s\in[n^i]\), and \(a\in[m^i]\). Then \(\{\Pi^i\}_{i=1}^p\) is a soft-Bellman equilibrium for \(\{\mathcal{M}^i\}_{i=1}^p\).
\end{theorem}

\begin{proof}
First, \(h(Y)\) is a concave function of matrix \(Y\) \cite{bloem2014infinite,zhou2017infinite}, and \(\ell^i(Y)\) is also a concave function of \(Y\) since \(C^{ii}\preceq 0\). Next, by applying the chain rule to \eqref{eqn: potential} and \eqref{eqn: entropy} we can show the following:
\begin{equation*}
     \textstyle \partial_{Y_{sa}}\ell^i(Y)= R_{sa}^i,\enskip \partial_{Y_{sa}} h(Y)  = \log\left(\sum_{j=1}^m Y_{sj}\right)-\log (Y_{sa}),
\end{equation*}
where \(R^i\in\mathbb{R}^{n^i\times m^i}\) satisfies \eqref{eqn: game reward}. Since \(Y\in\mathbb{R}^{n^i\times m^i}_{>0}\) is an optimal solution for optimization~\eqref{opt: best response}, there exists \(v^i\in\mathbb{R}^{n^i}\) such that the following Karush-Kuhn-Tucker conditions hold:\begin{subequations}
\begin{align}
& \sum\limits_{a=1}^{m^i}Y_{sa}^i= q_s^i + \gamma \sum\limits_{j=1}^{n^i}\sum\limits_{a=1}^{m^i} T_{jas}^i Y_{ja}^i,\\
& \log (Y_{sa}^i)-\log\left(\sum_{j=1}^{m^i} Y_{sj}^i\right)=Q_{sa}^i-v_s^i,\label{eqn: kkt dual}
\end{align}
\end{subequations}
for all \(s\in[n^i]\) and \(a\in[m^i]\), where \(Q_{sa}^i\) is given by \eqref{eqn: soft Q}. Let \(\Pi_{sa}^i\) be given by \eqref{eqn: Y2Pi}, then \eqref{eqn: kkt dual} implies that
\begin{subequations}
    \begin{align}
       \Pi_{sa}^i&=\exp(Q_{sa}^i-v_s^i),\label{eqn: pi a}\\
       1=\sum_{j=1}^{m^i} \Pi_{sj}^i&=\sum_{j=1}^{m^i}\exp(Q_{sj}^i-v_s^i), \label{eqn: pi b}
    \end{align}
\end{subequations}
for all \(s\in[n^i]\) and \(a\in[m^i]\). By combining \eqref{eqn: pi a} with \eqref{eqn: pi b} one can obtain the condition in \eqref{eqn: soft policy}. Finally, multiplying both sides of \eqref{eqn: pi b} by \(\exp(v_s^i)\) gives \eqref{eqn: soft V}, which completes the proof.
\end{proof}

\section{Forward Solution via Nonlinear Least-squares} 
% We now discuss how to compute a soft-Bellman equilibrium by solving a nonlinear least-squares problem, together with the existence and uniqueness of the solution. 
We now establish the existence and uniqueness of a soft-Bellman equilibrium, followed by a discussion on how to compute it by solving a nonlinear least-squares problem.
To this end, we introduce the following notation:
\begin{equation}
\begin{aligned}
    l& \coloneqq \sum_{i=1}^p m^in^i, \enskip b  \coloneqq \begin{bmatrix}
    (b^1)^\top & (b^2)^\top & \cdots & (b^p)^\top
    \end{bmatrix}^\top, \\
    r &\coloneqq \sum_{i=1}^p n^i,\enskip q  \coloneqq \begin{bmatrix}
    (q^1)^\top & (q^2)^\top & \cdots & (q^p)^\top
    \end{bmatrix}^\top.
\end{aligned} 
\end{equation}
Let matrices \(D^i, E^i\in\mathbb{R}^{n^i\times m^in^i}\) be such that
\begin{equation}\label{eqn: incidence game}
D^i= I_{n^i}\otimes (\mathbf{1}_{m^i}^\top), \enskip E^i_{kj} = T_{\quot(j, m)+1,\rem(j, m), k}^i,
\end{equation}
for all \(k\in[n^i]\) and \(j\in[m^in^i]\). Furthermore, let
\begin{equation}
\begin{aligned}
    H & \coloneqq \blkdiag(D^1-\gamma E^1, D^2-\gamma E^2, \ldots, D^p-\gamma E^p),\\
    K & \coloneqq \blkdiag( (D^1)^\top D^1,  (D^2)^\top D^2, \ldots, (D^p)^\top D^p),\\
    C & \coloneqq \begin{bmatrix}
    C^{11} & C^{12} & \cdots & C^{1p}\\
    C^{21} & C^{22} & \cdots & C^{2p}\\
    \vdots & \vdots & \ddots & \vdots\\
    C^{p1} & C^{p2} & \cdots & C^{pp}\\
    \end{bmatrix}.
\end{aligned}
\end{equation}

With these notations, we are ready to establish the following results on the existence and uniqueness of the soft-Bellman equilibrium. 

\begin{theorem}\label{thm: existence & uniqueness}
    Let \(\{\mathcal{M}^i=\{[n^i], [m^i], q^i, T^i, R^i, \gamma\}\}_{i=1}^p\) be an affine Markov game where \(C^{ii}\preceq 0\) for all \(i\in[p]\). Then \(Y^i\in\mathbb{R}^{n^i\times m^i}_{>0}\) is an optimal solution of optimization~\eqref{opt: best response} for all \(i\in[p]\) if and only if there exists \(v\in\mathbb{R}^r\) such that
    \begin{equation}\label{eqn: nl eqn}
        \begin{aligned}
         \log(y) &= \log(K y) +b+Cy-H^\top v,\\
             Hy&=q,
        \end{aligned}
    \end{equation}
    where 
    \begin{equation}
        y = \begin{bmatrix}
        \vect(Y^1)^\top & \vect(Y^2)^\top & \cdots & \vect(Y^p)^\top
        \end{bmatrix}^\top.
    \end{equation}
    Furthermore, there exists \(y\in\mathbb{R}^{l}\) such that \eqref{eqn: nl eqn} holds for some \(v\in\mathbb{R}^{r}\). If \(C+C^\top\preceq 0\), then such a \(y\) is unique.
\end{theorem}

\begin{proof}
First of all, the conditions \eqref{eqn: nl eqn} are the union of the KKT conditions for optimization~\eqref{opt: best response} for all \(i\in[p]\). Due to the assumption that \(C^{ii}\preceq 0\), \(C+C^\top \preceq 0\), and the strict concavity of logarithm function, one can verify that \(Y^i\in\mathbb{R}^{m^i\times n^i}\) is an optimal solution of optimization~\eqref{opt: best response} for all \(i\in[p]\) if and only if \(\{Y^i\}_{i=1}^p\) is a Nash equilibrium of a \(p\)-player diagonally strictly concave game, which exists and is unique \cite{rosen1965existence}.  
\end{proof}

As a result of \Cref{thm: existence & uniqueness}, one can compute a soft-Bellman equilibrium by solving the following nonlinear least-squares problem:
\begin{equation}\label{opt: nl ls}
\begin{array}{ll}
    \underset{y, v}{\mbox{minimize}} & \norm{\log(Ky)+b+Cy-H^\top v-\log(y)}^2\\&+\norm{Hy-q}^2
\end{array}
\end{equation}
Notice that the optimal value of the above optimization is zero, since there exists at least one solution for the nonlinear equations in \eqref{eqn: nl eqn}. 

\section{Inverse Learning via Implicit Differentiation}

Given the parameters of an affine Markov game, one can compute a soft-Bellman equilibrium of this game by solving the nonlinear least-squares problem in \eqref{opt: nl ls}. The question remains, however, of how to infer these parameters such that they best explains observed decisions, a problem also known as the \emph{inverse game}. Next, we answer this question by developing a projected-gradient method for parameter calibration.  

The inverse game problem is a parameter optimization problem defined as follows. We start with a set of empirically observed equilibrium state-action frequencies, 
\begin{equation}
    \hat{Y}^1, \hat{Y}^2, \ldots, \hat{Y}^p,
\end{equation}
where \(\hat{Y}^i_{sa}\in\mathbb{R}^{m^i\times n^i}\) denotes the empirical probability for player \(i\) to choose action \(a\) in state \(s\). Let 
\begin{equation}\label{eqn: observation}
    \hat{y}\coloneqq \begin{bmatrix}
    \vect(\hat{Y}^1)^\top & \vect(\hat{Y}^2)^\top & \cdots & \vect(\hat{Y}^p)^\top
    \end{bmatrix}^\top. 
\end{equation}
To find the best parameters that explain the observed state-action frequency matrices in \eqref{eqn: observation}, one can solve the following optimization problem
\begin{equation}\label{opt: bilevel}
    \begin{array}{ll}
    \underset{y, v, b, C}{\mbox{minimize}} &  \norm{y-\hat{y}}^2\\
    \mbox{subject to} &  \log(y) = \log(K  y) +b+Cy-H^\top v,\\
             & Hy=q, \enskip b\in\mathbb{B}, \enskip C\in\mathbb{D},
    \end{array}
\end{equation}
where \(\mathbb{B}\subset\mathbb{R}^l\) and \(\mathbb{D}\subset\mathbb{R}^{l\times l}\) are closed convex constraint sets for vector \(b\) and matrix \(C\), respectively.

Solving problem \eqref{opt: bilevel} is numerically challenging because this optimization contains both nonlinear equation constraints and possible positive semi-definite cone constraints in set \(\mathbb{D}\). As a remedy, we propose an approximate projected-gradient method that combines nonlinear equation solving with efficient projections. To this end, we let
\begin{equation}
    J=\begin{bmatrix}
    K\diag(Ky)^{-1} -\diag(y)^{-1}+C & -H^\top\\
    H & 0_{r\times r}.
    \end{bmatrix}
\end{equation}

By using the chain rule and the implicit function theorem \cite{dontchev2014implicit} one can show that, if \eqref{eqn: nl eqn} holds and matrix \(J\) is nonsingular, then 
\begin{subequations}
\begin{align}
    \partial_b \norm{y-\hat{y}}^2&=-2\begin{bmatrix}
    (y-\hat{y})^\top & 0_{r}
    \end{bmatrix}J^{-1} \begin{bmatrix}
    I_l\\
    0_{r\times l}
    \end{bmatrix},\\
    \partial_{C_j} \norm{y-\hat{y}}^2&=-2y_j\begin{bmatrix}
    (y-\hat{y})^\top & 0_{r}
    \end{bmatrix}J^{-1} \begin{bmatrix}
    I_l\\
    0_{r\times l}
    \end{bmatrix},
\end{align}
\end{subequations}
for all \(j\in[l]\), where \(C_j\in\mathbb{R}^l\) is the \(j\)-th column of matrix \(C\). Hence one can compute the approximate gradient for vector \(b\) and matrix \(C\) that locally decreases the value of the objective function in \eqref{opt: bilevel} as follows:
\begin{subequations}\label{eqn: approx grad}
    \begin{align}
        \tilde{\nabla}_b \norm{y-\hat{y}}^2 & \coloneqq -2\begin{bmatrix}
    I_l & 0_{l\times r}
    \end{bmatrix}(J^\dagger)^\top \begin{bmatrix}
    y-\hat{y}\\
    0_{r}
    \end{bmatrix},\label{eqn: b grad}\\
    \tilde{\nabla}_C \norm{y-\hat{y}}^2 & \coloneqq -2\begin{bmatrix}
    I_l & 0_{l\times r}
    \end{bmatrix}(J^\dagger)^\top \begin{bmatrix}
    y-\hat{y}\\
    0_{r}
    \end{bmatrix}y^\top.\label{eqn: C grad} 
    \end{align}
\end{subequations}
Notice that we approximate \(J^{-1}\) with the Moore-Penrose pseudoinverse \(J^\dagger\) in \eqref{eqn: approx grad}. Such an approximation is exact if \(J\) is nonsingular, and still well-defined even if \(J\) is singular or \(J^{-1}\) is numerically unstable to compute.

Based on the formulas in \eqref{eqn: approx grad}, we propose an approximate projected-gradient method, summarized in Algorithm~\ref{alg: proj grad}, to solve optimization~\eqref{opt: bilevel}, where we let
\begin{subequations}
\begin{align}
    \text{Proj}_{\mathbb{B}}(b)  & \coloneqq \underset{z\in\mathbb{B}}{\argmin} \norm{z-b},\\
    \text{Proj}_{\mathbb{D}}(C)  & \coloneqq \underset{X\in\mathbb{D}}{\argmin} \norm{X-C}_F,
\end{align}  
\end{subequations}
for all \(b\in\mathbb{R}^l\) and \(C\in\mathbb{R}^{l\times l}\). 
Each iteration of this method first solves the nonlinear least-squares problem in \eqref{opt: nl ls}, then performs a projected-gradient step on \(b\) and \(C\).  

\begin{algorithm}[!ht]
\caption{Approximated projected-gradient method.}
\begin{algorithmic}[1]
\Require  Step size \(\alpha\in\mathbb{R}_{>0}\), number of iterations \(k_{\max}\in\mathbb{N}\), random initial parameters $b_{init}\in\mathbb{R}^{l}, C_{init}\in\mathbb{R}^{l\times l}$, tolerance \(\epsilon\in\mathbb{R}\).
\State Initialize \(k=1\), \(b=b_{init}\), \(C=C_{init}\).
\While{\(k<k_{\max}\)}
\State Solve optimization \eqref{opt: nl ls} for \(y\).
\If{change in $\norm{y-\hat{y}}^2 < \epsilon$}
\State terminate.
\EndIf
\State \(b\gets \text{Proj}_{\mathbb{B}}(b-\alpha \tilde{\nabla}_b\norm{y-\hat{y}}^2)\) \Comment{cf. \eqref{eqn: b grad}} \label{lst:line:b_update}
\State \(C\gets \text{Proj}_{\mathbb{D}}(C-\alpha \tilde{\nabla}_C\norm{y-\hat{y}}^2)\) \Comment{cf. \eqref{eqn: C grad}} \label{lst:line:C_update}
\State \(k\gets k+1\)
\EndWhile
\Ensure Vector \(b\) and matrix \(C\).
\end{algorithmic}
\label{alg: proj grad}
\end{algorithm}

\section{Experiments}
\label{sec: experiments}
We evaluate the performance of the proposed algorithm against a baseline algorithm that neglects the fact that players' reward functions depend upon each other's actions in a predator-prey OpenAI Gym environment \cite{magym}. 
% We compare our algorithm with a baseline algorithm that applies IRL for each player individually. Results show that by considering coupling between players, the reward parameters inferred by the proposed algorithm outperform those inferred by the baseline by reducing the Kullback-Leibler divergence between the equilibrium policies and observed policies by at least two orders of magnitude.
We solve the forward problem by specifying the nonlinear least-squares problem \eqref{opt: nl ls} in Julia \cite{bezanson2017julia} using the JuMP \cite{dunning2017jump} interface and the COIN-OR IPOPT \cite{wachter2006implementation} optimizer. The source code is publicly available at \url{https://github.com/vivianchen98/Inverse_MDPGame}.

\subsection{Baseline}
The baseline algorithm is a decoupled version of \Cref{alg: proj grad}, that is, it solves optimization \eqref{opt: best response} with the coupling parameter $C^{ij}=0$ for all players $i,j\in [p]$.
Dropping this parameter frees the baseline algorithm to solve an optimization for each player independently, similar to many existing multi-agent inverse reinforcement learning algorithms. 

% That is, for $s\in[n]$, each player solves
% \begin{equation}
%     \begin{array}{ll}
%  \underset{Y\in\mathbb{R}^{n\times m}}{\mbox{maximize}}  & \vect(Y)^\top b +h(Y)\\
%       \mbox{subject to}  & \sum\limits_{a=1}^{m}Y_{sa}=p_s + \gamma \sum\limits_{j=1}^{n}\sum\limits_{a=1}^{m} T_{jas} Y_{ja}
%       \end{array}
% \end{equation} 

% \begin{figure}[h]
%   \centering
%   \begin{subfigure}[b]{0.49\linewidth}
%     \includegraphics[width=\linewidth]{figures/predator_prey_env.pdf}
%     \caption{Predator-prey environment.}
%     \label{fig:figure1}
%   \end{subfigure}
%   \hfill
%   \begin{subfigure}[b]{0.49\linewidth}
%     \includegraphics[width=\linewidth]{figures/test score.pdf}
%     \caption{Test return of MARL models.}
%     \label{fig:figure2}
%   \end{subfigure}
%   \caption{Environment and Data Collection}
%   \label{fig:both}
% \end{figure}

\subsection{Algorithm Parameters}
For the projected-gradient method in \Cref{alg: proj grad}, we use a backtracking line search technique to fine-tune the step size in line~\ref{lst:line:b_update} and \ref{lst:line:C_update} based upon the Armijo (sufficient decrease) condition \cite{armijo1966minimization}. Each iteration starts with an initial step size $\alpha=1$, and the algorithm reduces the step size by half until it meets the sufficient decrease condition. Both algorithms terminate when the change in $\norm{y-\hat{y}}^2$ is below a given tolerance $\epsilon=0.005$. The maximum number of iterations $k_{\max}$ is $100$, and the discount factor $\gamma$ is $0.99$. We sample the values of the vector $b_{init}$ and the matrix $C_{init}$ from a random number generator given a seed. We run both algorithms from seed $1$ to $10$. 

\subsection{Predator-Prey Environment}
\label{subsec: predator_prey_env}
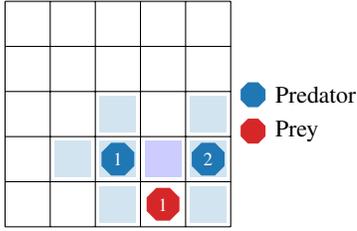
\begin{figure}
    \centering
    \begin{tikzpicture}[scale=0.6]
  % Draw the 5x5 grid
  \draw[step=1cm,color=black] (0,0) grid (5,5);

  % Fill light blue cells
  \fill[predatorblue!20] (1.1,1.1) rectangle (1.9,1.9);
  \fill[predatorblue!20] (2.1,1.1) rectangle (2.9,1.9);
  \fill[predatorblue!20] (2.1,0.1) rectangle (2.9,0.9);
  \fill[predatorblue!20] (2.1,2.1) rectangle (2.9,2.9);

  \fill[blue!20] (3.1,1.1) rectangle (3.9,1.9);

  \fill[predatorblue!20] (4.1,0.1) rectangle (4.9,0.9);
  \fill[predatorblue!20] (4.1,1.1) rectangle (4.9,1.9);
  \fill[predatorblue!20] (4.1,2.1) rectangle (4.9,2.9);
  
  % Draw the blue octagon nodes in two cells
  \node[fill=predatorblue, regular polygon, regular polygon sides=8, inner sep=2pt] at (2.5,1.5) {\scriptsize\textcolor{white}{1}};
  \node[fill=predatorblue, regular polygon, regular polygon sides=8, inner sep=2pt] at (4.5,1.5) {\scriptsize\textcolor{white}{2}};

  % Draw the red octagon node in one cell
  \node[fill=preyred, regular polygon, regular polygon sides=8, inner sep=2pt] at (3.5,0.5) {\scriptsize\textcolor{white}{1}};
    
    % Add the legend
    \matrix [right] at (current bounding box.east) {
      \node [fill=predatorblue, regular polygon, regular polygon sides=8,label=right:{\small Predator}] {}; \\
      \node [fill=preyred, regular polygon, regular polygon sides=8,label=right:{\small Prey}] {}; \\
    };
\end{tikzpicture}
    \caption{Two predators (blue) and one prey (red) moving in a 5x5 GridWorld. The light blue cells represent the catching region of the predators, and the light purple cell represents the overlapping of both predators' catching regions. This episode terminates when the prey is inside a light purple cell.}
    \label{fig:predator_prey_env}
\end{figure}

We consider a predator-prey environment from a collection of multi-agent environments based on OpenAI Gym \cite{magym}. As shown in \cref{fig:predator_prey_env}, two predators attempt to capture one randomly moving prey in a $5\times 5$ GridWorld. Each predator has observations of all players and the coordinates of the prey relative to itself and selects one of five actions: \texttt{left}, \texttt{right}, \texttt{up}, \texttt{down}, or \texttt{stop}. The prey is caught when it is within the catching region (light blue cells in \cref{fig:predator_prey_env}) of at least one predator. An episode terminates when the prey is caught by more than one predator (inside a light purple cell in \cref{fig:predator_prey_env}), resulting in a positive reward. For every new episode, the environment initializes the prey into random locations and the prey never moves voluntarily into the predators' neighborhood.
In this environment, only the two predators are controllable, but we collect the trajectories of all three players, including the prey, to solve the inverse game problem.

\subsection{Observed Dataset Collection}
We collect all players' trajectories as the observed interactions. Each trajectory is a sequence of states and actions until termination for the current episode. 
We train a policy using a multi-agent reinforcement learning algorithm \cite{sunehag2017value} and sample trajectories from this policy. The players in this policy exhibit uncertainties in their decision-making process that are difficult to articulate explicitly, much like humans. As a result, the data from these models can serve as a proxy for human datasets.
% \david{this sentence needs revision}

We process the collected trajectories from all three players by first pruning those shorter than the $50$th percentile of trajectory lengths and then capping the remaining trajectories to the same length. After processing, we attain $100$ useful trajectories of length $6$. We compute the collection of state-action frequencies for all three players $\hat{y}$ and approximate the initial state distributions and the transition probabilities for all players using the observed data.

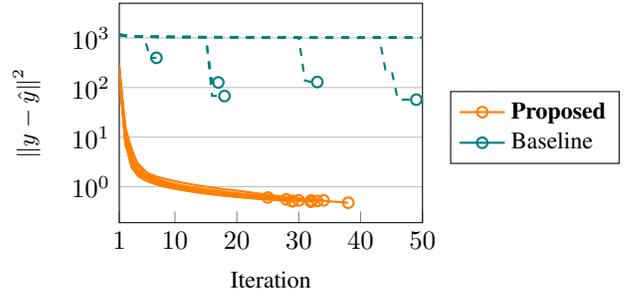
\begin{figure}
    \centering
    \begin{tikzpicture}
    \begin{semilogyaxis}[
            xlabel={{\small Iteration}},
            ylabel={{\small $\norm{y-\hat{y}}^2$}},
            xmin=1, xmax=50,
            ymin=0, ymax=5000,
            width=0.65\linewidth,
            height =4.5cm, 
            xtick={1,10,20,30,40,50},
            xtick pos=left,
            ytick={1,10,100,1000},
            ytick style={draw=none},
            legend style={at={(1.65,0.6)},anchor=north east, legend cell align={left}},
            ymajorgrids = true,
            % xmajorgrids = true
    ]

    % Custom legend symbol
    \addlegendimage{orange,mark=o, thick}
    \addlegendentry{{\small \bf Proposed}}

    \addlegendimage{teal,mark=o, thick}
    \addlegendentry{{\small Baseline}}
    
    % game terminations
    \addplot[only marks, mark=o, orange, thick]coordinates {(34,0.532066938000832) (32,0.4992389546) (30,0.5306287569) (28,0.5563287063) (33,0.5180440768) (29,0.5098950252) (38,0.4764974208) (32,0.5275647948) (29,0.5169438494) (25,0.6065868245)};
    
    % decoupled terminations
    \addplot[only marks, mark=o, teal, thick]coordinates {(7,398.3584228383) (79,1025.82921208205) (100,1026.258229) (49, 57.06564416) (55,1025.238772) (18, 67.55653781) (17,126.8949034) (83,1025.054204) (59,1024.884266) (33,129.5089324)};

    \addplot [mark=none, orange, thick] table [x=iter, y=seed1, col sep=comma]{psi_game.dat};
    \addplot [mark=none, orange, thick] table [x=iter, y=seed2, col sep=comma]{psi_game.dat};
    \addplot [mark=none, orange, thick] table [x=iter, y=seed3, col sep=comma]{psi_game.dat};
    \addplot [mark=none, orange, thick] table [x=iter, y=seed4, col sep=comma]{psi_game.dat};
    \addplot [mark=none, orange, thick] table [x=iter, y=seed5, col sep=comma]{psi_game.dat};
    \addplot [mark=none, orange, thick] table [x=iter, y=seed6, col sep=comma]{psi_game.dat};
    \addplot [mark=none, orange, thick] table [x=iter, y=seed7, col sep=comma]{psi_game.dat};
    \addplot [mark=none, orange, thick] table [x=iter, y=seed8, col sep=comma]{psi_game.dat};
    \addplot [mark=none, orange, thick] table [x=iter, y=seed9, col sep=comma]{psi_game.dat};
    \addplot [mark=none, orange, thick] table [x=iter, y=seed10, col sep=comma]{psi_game.dat};

    \addplot [mark=none, dashed, teal, thick] table [x=iter, y=seed1, col sep=comma]{psi_decoupled.dat};
    \addplot [mark=none, dashed, teal, thick] table [x=iter, y=seed2, col sep=comma]{psi_decoupled.dat};
    \addplot [mark=none, dashed, teal, thick] table [x=iter, y=seed3, col sep=comma]{psi_decoupled.dat};
    \addplot [mark=none, dashed, teal, thick] table [x=iter, y=seed4, col sep=comma]{psi_decoupled.dat};
    \addplot [mark=none, dashed, teal, thick] table [x=iter, y=seed5, col sep=comma]{psi_decoupled.dat};
    \addplot [mark=none, dashed, teal, thick] table [x=iter, y=seed6, col sep=comma]{psi_decoupled.dat};
    \addplot [mark=none, dashed, teal, thick] table [x=iter, y=seed7, col sep=comma]{psi_decoupled.dat};
    \addplot [mark=none, dashed, teal, thick] table [x=iter, y=seed8, col sep=comma]{psi_decoupled.dat};
    \addplot [mark=none, dashed, teal, thick] table [x=iter, y=seed9, col sep=comma]{psi_decoupled.dat};
    \addplot [mark=none, dashed, teal, thick] table [x=iter, y=seed10, col sep=comma]{psi_decoupled.dat};
    \end{semilogyaxis}
    \end{tikzpicture}
    \caption{Algorithms for the inverse game problem with termination marked in circles (the lower the better).}
    \label{fig:psi_convergence}
\end{figure}

% pi gridworld 3x2
\begin{figure}[h]
\begin{tabular}{p{0.1cm}cc}
    & {\small \bf \centering Proposed} & {\small Baseline}\\
    
    \rotatebox[origin=c]{90}{{\small Predator} \tiny\octagoned{1}{predatorblue}{white}}  & \includegraphics[width=0.44\linewidth,valign=m]{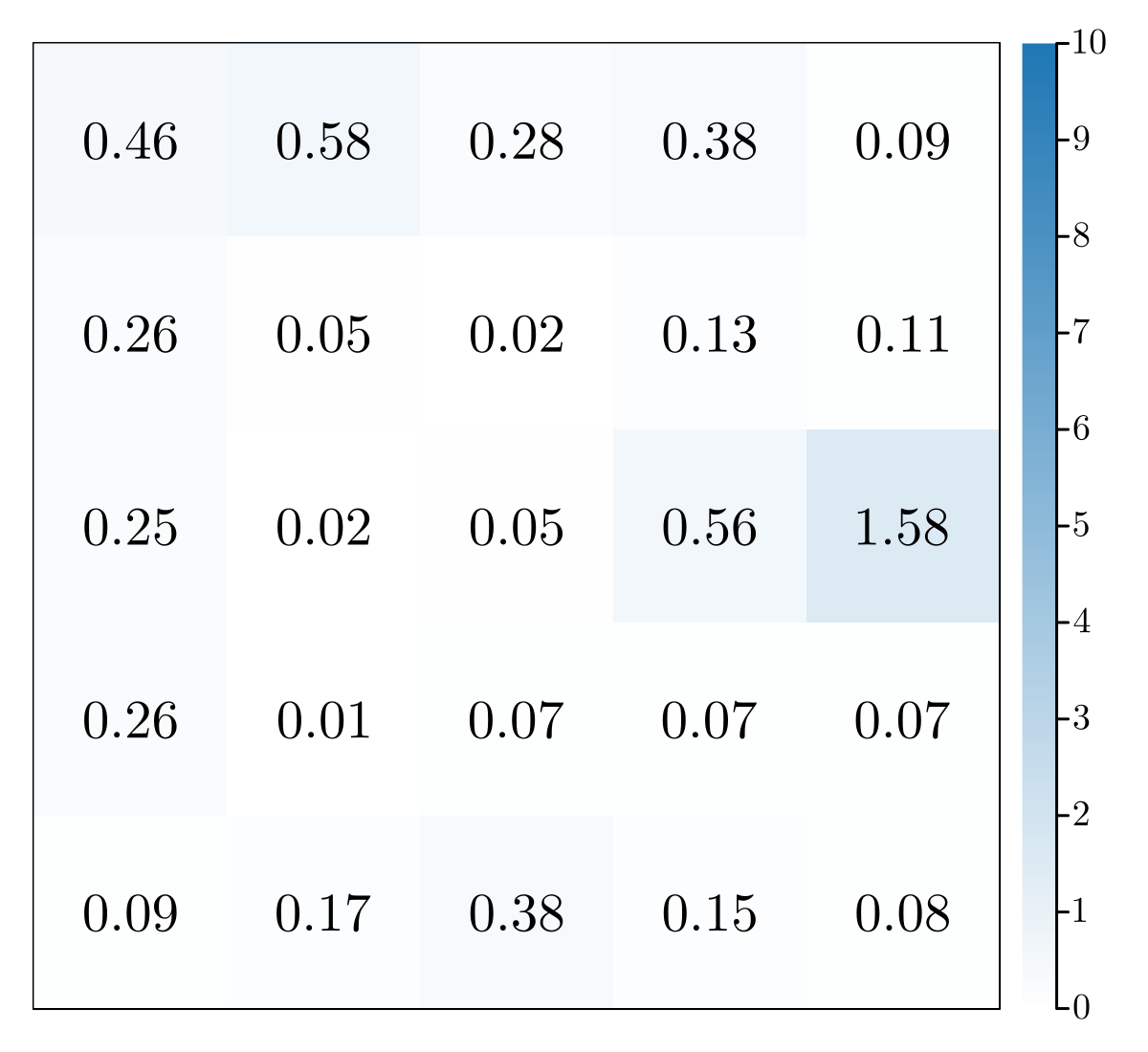} & \includegraphics[width=.44\linewidth,valign=m]{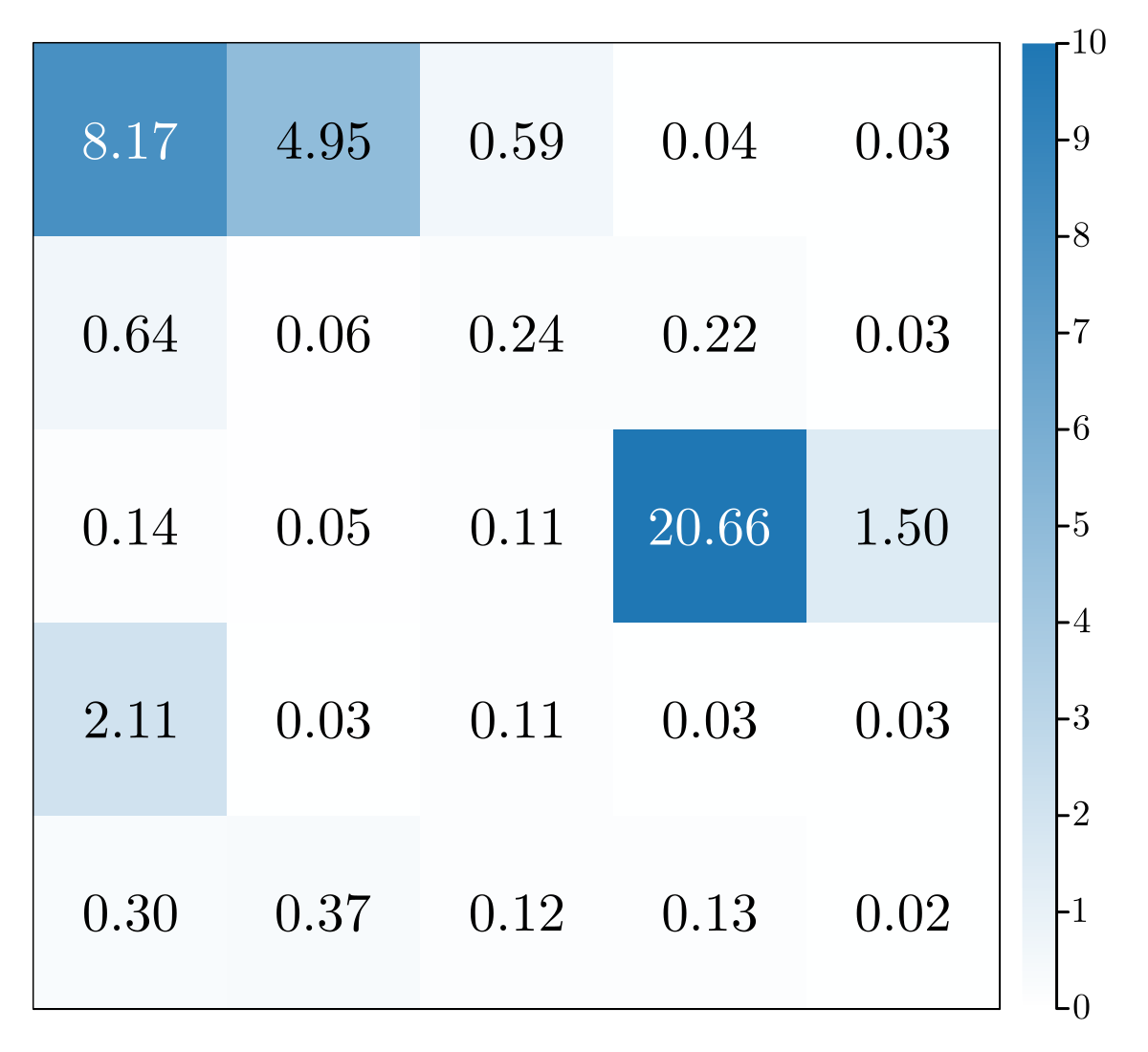}\\
    
    \rotatebox[origin=c]{90}{{\small Predator} \tiny\octagoned{2}{predatorblue}{white}} & \includegraphics[width=.44\linewidth,valign=m]{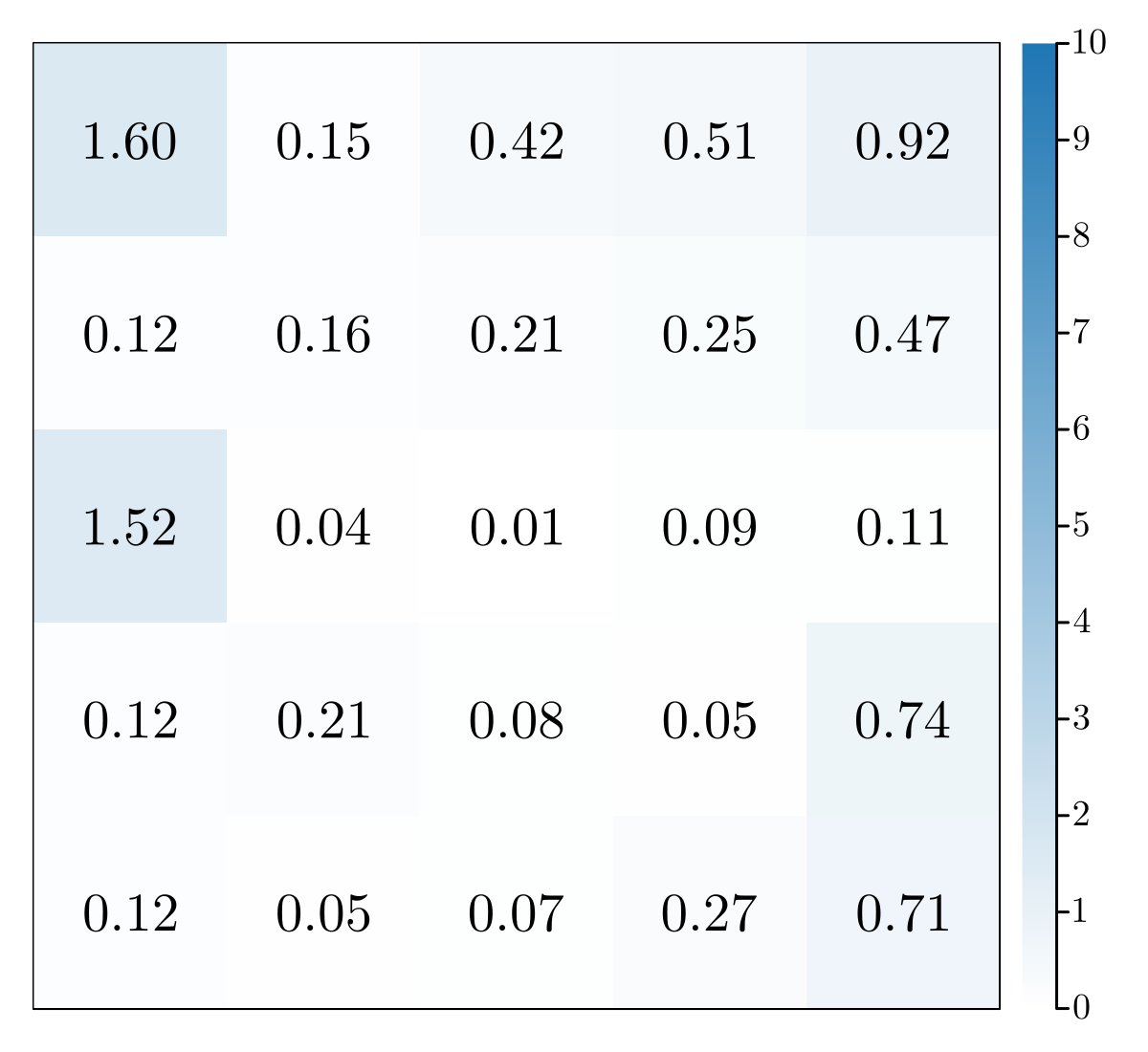} & \includegraphics[width=.44\linewidth,valign=m]{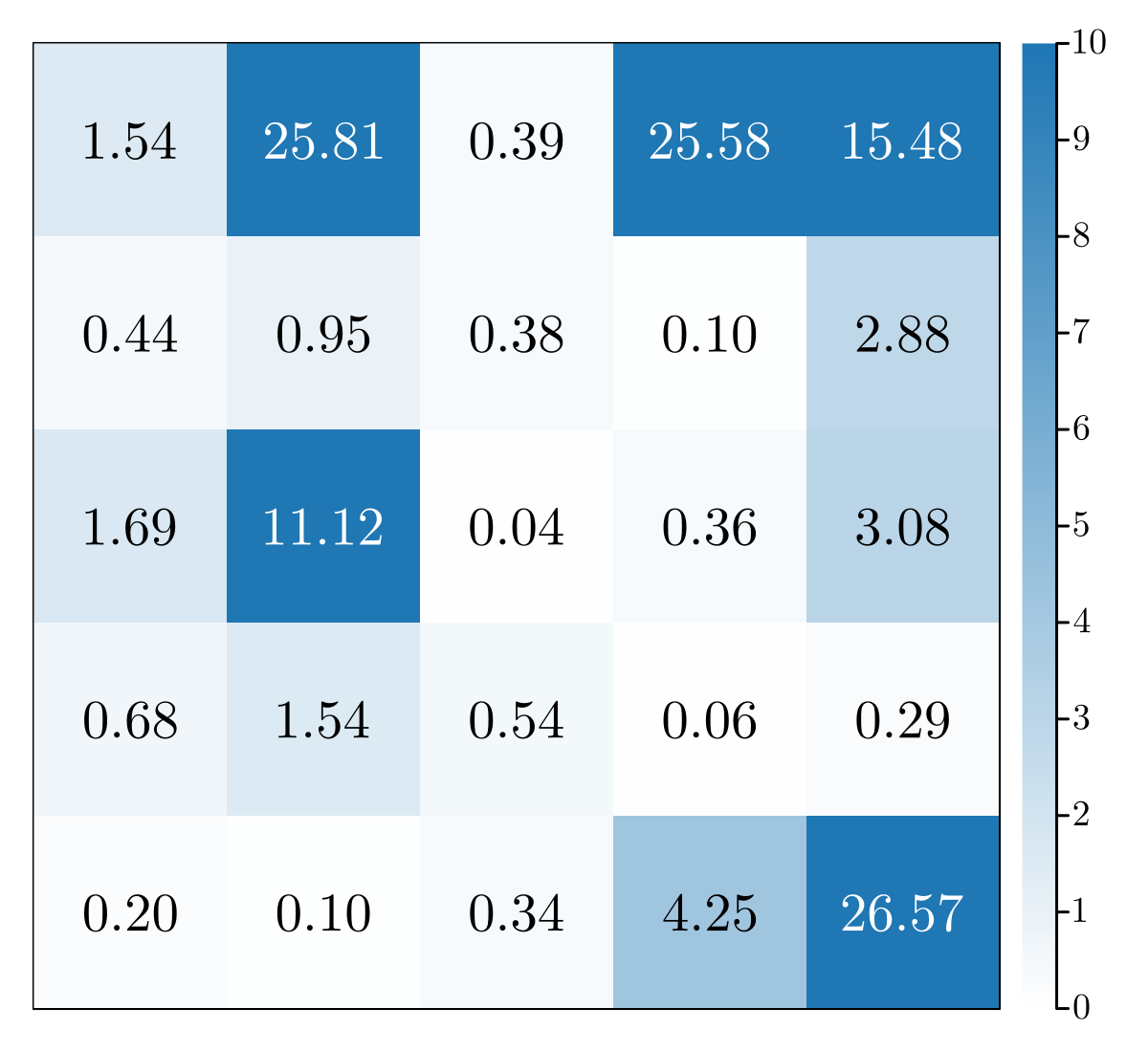}\\
    
    \rotatebox[origin=c]{90}{{\small Prey} \tiny\octagoned{1}{preyred}{white}} & \includegraphics[width=.44\linewidth,valign=m]{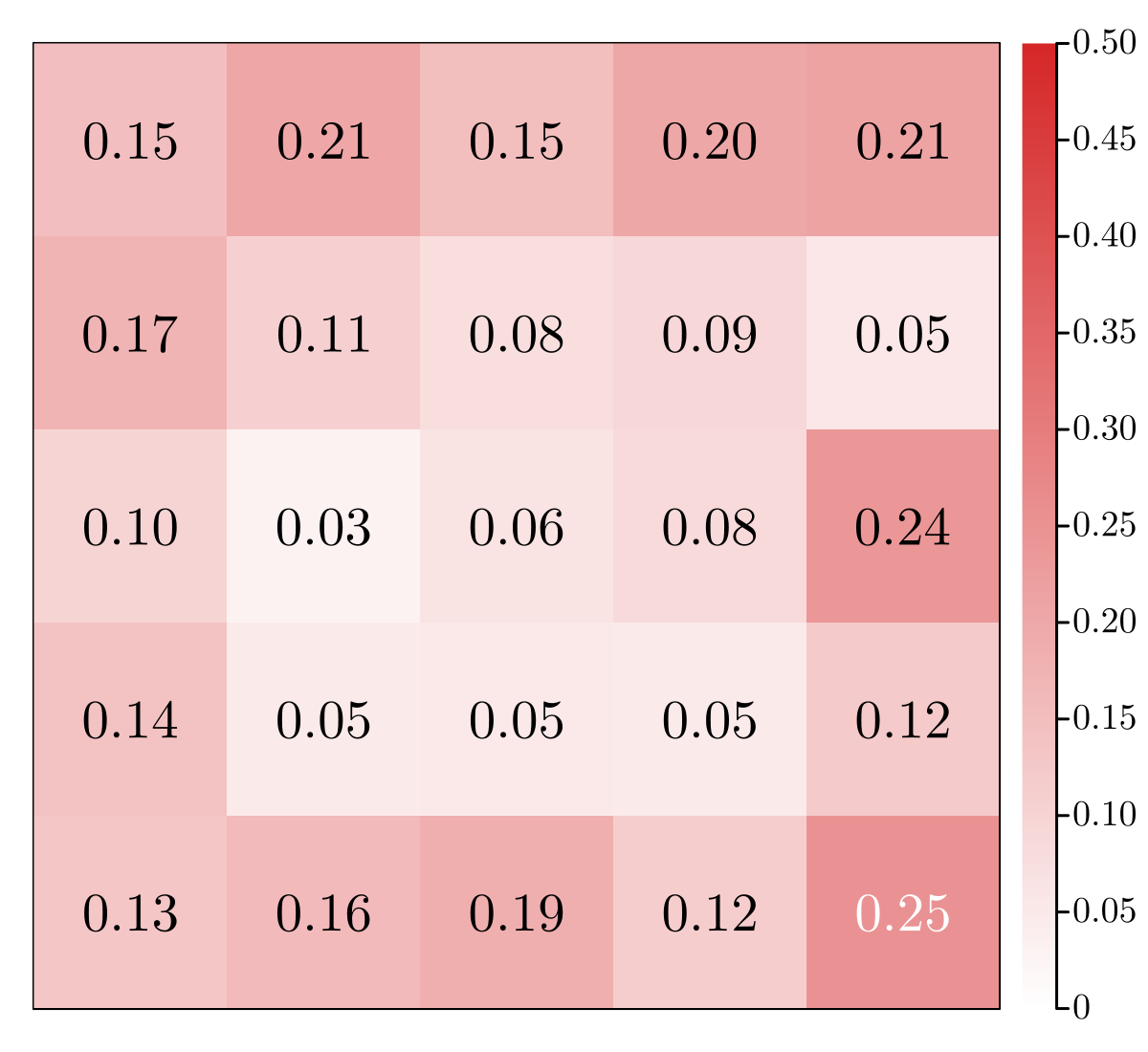} & \includegraphics[width=.44\linewidth,valign=m]{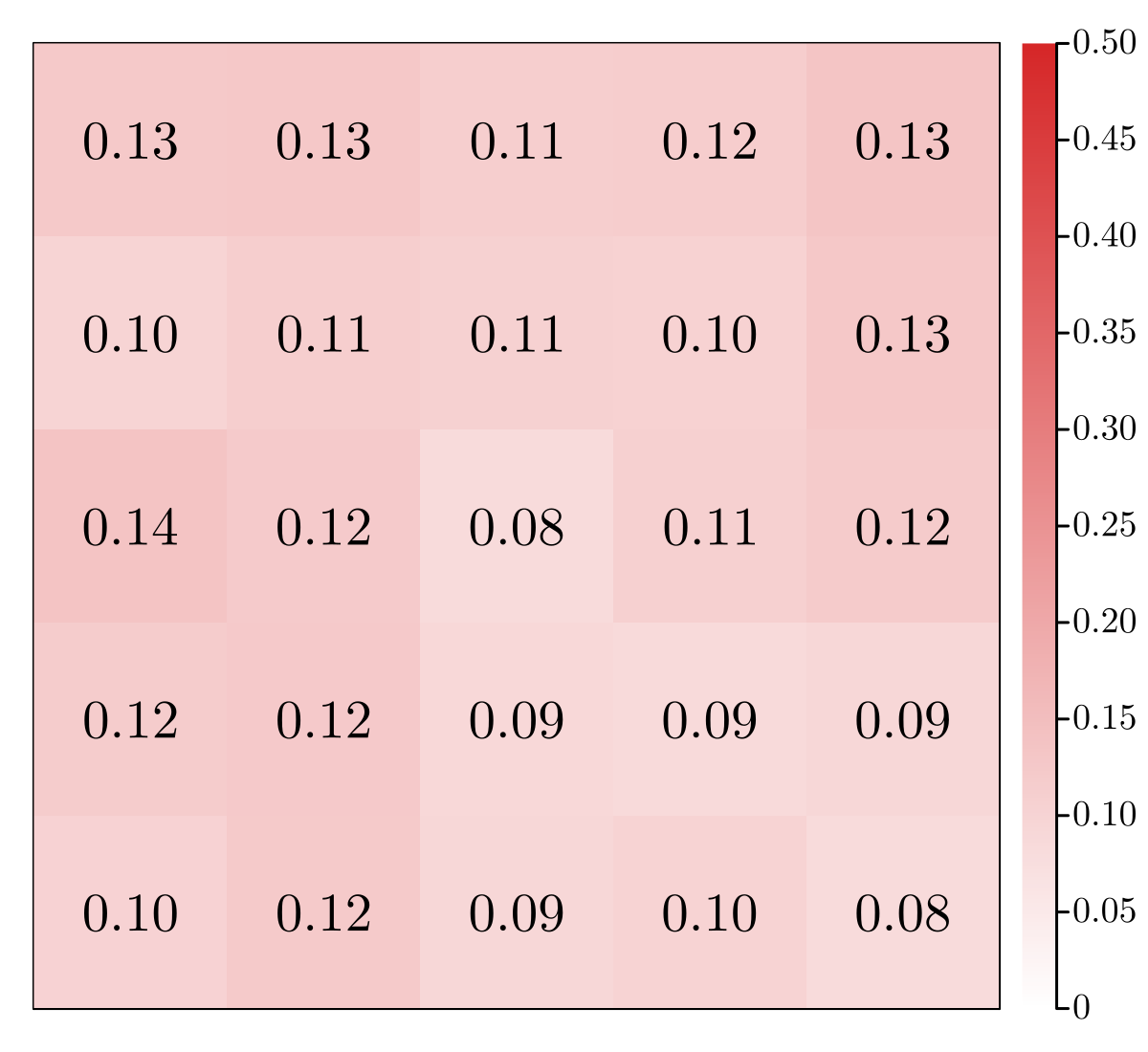}\\
\end{tabular}
\caption{Heatmaps showing Kullback–Leibler divergence $D_{\text{KL}}(\Pi^i_s \parallel \hat{\Pi}^i_s)$ between the equilibrium policy $\Pi^i_s$ and the observed policy $\hat{\Pi}^i_s$ at each state $s$ in the GridWorld for all three players. All values rounded to two decimal places, the smaller (lighter color) the better.}
\label{fig:kl_gridworld}
\end{figure}

\subsection{Numerical Results}
We demonstrate \cref{alg: proj grad} and the baseline algorithm on the predator-prey environment introduced in \cref{subsec: predator_prey_env}. \cref{fig:psi_convergence} shows $\norm{y-\hat{y}}^2$, the squared norm of the difference between the computed state-action frequency $y$ and the observed state-action frequency matrices $\hat{y}$, with respect to the number of iterations. 
% Throughout we let initial step size be $\alpha=1$, maximum number of iterations be $k_{\max}=100$, tolerance be $\delta = 0.005$, and the discount factor be $\gamma=0.99$. We run both the proposed algorithm and the baseline for ten different initializations from seeds $1$ to $10$. 
Results show the proposed algorithm ends in $31.0\pm 3.6$ iterations, while the baseline algorithm takes $50.0\pm 31.3$ iterations to terminate. As shown in \cref{fig:psi_convergence}, the final iterate of the proposed algorithm has $\norm{y-\hat{y}}^2$ below $1$, while the baseline algorithm on average terminates with a value above $590.7$. This comparison highlights the importance of accounting for the coupling between the players.

Given a state-action frequency matrices $y^i$ for player $i$, we compute the corresponding policies $\Pi^i_{sa}$ by \eqref{eqn: Y2Pi}, and denote the equilibrium policy at each state $s$ as a probability distribution
$$\Pi^i_s = \begin{bmatrix}\Pi^i_{s,\texttt{left}} & \Pi^i_{s,\texttt{right}} & \Pi^i_{s,\texttt{up}} & \Pi^i_{s,\texttt{down}} & \Pi^i_{s,\texttt{stop}}\end{bmatrix}.$$
We report the Kullback–Leibler divergence $D_{\text{KL}}(\Pi^i_s \parallel \hat{\Pi}^i_s)$ between the equilibrium policy $\Pi^i_s$, computed using the proposed and the baseline algorithms, and the observed policy $\hat{\Pi}^i_s$ at each state $s$ for all three players. \cref{fig:kl_gridworld} shows that \cref{alg: proj grad} arrives at an equilibrium policy closer to the observed policy than the baseline algorithm does.

\section{Conclusion \& Future Work}
\label{sec: conclusion}

We proposed soft-Bellman equilibrium as a novel solution concept in affine Markov games, a class of Markov games where an affine reward function couples the players' actions, to capture interactions of boundedly rational players in stochastic, dynamic environments. We provided conditions for the existence and uniqueness of the soft-Bellman equilibrium. We solved the forward problem of computing such an equilibrium for a given affine Markov game and proposed an algorithm to tackle the inverse game problem of inferring players' reward parameters from observed interactions.

Future work should validate the effectiveness of the proposed algorithms using human datasets instead of synthetic datasets. For example, the INTERACTION dataset contains human driving trajectories in interactive traffic scenes \cite{zhan2019interaction}, and can serve as a more representative dataset for the inverse game problem.


%%%%%%%%%%%%%%%%%%%%%%%%%%%%%%%%%%%%%%%%%%%%%%%%%%%%%%%%%%%%%%%%%%%%%%%%%%%%%%%%

%\section*{APPENDIX}
\section*{ACKNOWLEDGMENT} The authors would like to thank Negar Mehr and Xiao Xiang for their constructive feedback.

% \bibliographystyle{IEEEtran}
% \bibliography{IEEEabrv,reference}
\printbibliography

\end{document}